\documentclass[11pt]{article}

\usepackage{amsmath,amssymb,amsthm}

\newcommand{\ds}{\displaystyle}

\newtheorem{theorem}{Theorem}
\newtheorem{proposition}{Proposition}
\newtheorem{Def}{Definition}

\newcommand{\R}{\mathbb{R}}

\begin{document}

\title{
On recoverability properties of fixed measurement matrices
}
\author{Anatoly Eydelzon\\
Department of Mathematical Sciences\\
The University of Texas at Dallas\footnote{anatoly@utdallas.edu}}
\maketitle

\begin{abstract}
The purpose of this paper is to extend a result by Donoho and Huo, Elad and Bruckstein,  Gribnoval and Nielsen on sparse representations
of signals in dictionaries to general matrices. We consider a general fixed measurement matrix, not necessarily a dictionary, and derive
sufficient condition for having unique sparse representation of signals in this matrix. Currently, to the best of our knowledge, no such method exists.
In particular, if matrix is a dictionary, our method is at least as good as the method proposed by Gribnoval and Nielsen.
 
\end{abstract}

\section{Introduction}
Given a data vector $\tilde{x}\in\mathbb{R}^n$, the linear measurements $y_i$ of the data $\tilde{x}$ consist of the inner products of $\tilde{x}$ with a number of measurement vectors $a_i\in\mathbb{R}^n$, $i = 1,2,\ldots,m$, that is $y_i = \langle a_i,\tilde{x}\rangle$. In matrix form $\tilde{y} = A\tilde{x}$, where $A$ is an $m\times n$ matrix, called the measurement or encoding matrix, that consists of $a_i$'s as its rows and $m$ is the number of measurements. 

If the number of measurements is less than the dimension of the data, that is, $m<n$, the linear system $A\tilde{x} = \tilde{y}$ is under-determined, and therefore has infinitely many solutions, which makes the recovery of $\tilde{x}$ impossible. However if (a) the data vector $\tilde{x}$ is sufficiently sparse and (b) the encoding matrix $A$ contains a sufficient number of measurements and satisfies certain properties, then $\tilde{x}$ can be recovered (exactly or to a given accuracy) at a polynomial time complexity.

We consider the following recovery problem of a sparse vector $\tilde{x}\in\mathbb{R}^n$
from its linear measurement $\tilde{y} = A\tilde{x}\in\mathbb{R}^m$,
where $A$ is a known $m\times n$ full rank matrix and $m<n$. The associated optimization problem could be stated as
\begin{equation}\label{intro-0P}
\min_{x\in\mathbb{R}^n}\{\|x\|_0: Ax = \tilde{y}\},
\end{equation}
where $\|x\|_0$ is the number of nonzero entries of $x$. This problem is non-convex and therefore can not be solved by
conventional optimization methods.

On the other hand we can solve the following problem which can be written as a linear program (LP) via a standard transformation,
\begin{equation}\label{intro-LP}
\min_{x\in\mathbb{R}^n}\{\|x\|_1: Ax = \tilde{y}\}
\end{equation}
and ask a question: Under what conditions on A and $\tilde{x}$ are the problems (\ref{intro-0P}) and (\ref{intro-LP}) uniquely solved by $\tilde{x}$?

\begin{Def}[Partition]
By a partition (S,Z) we mean a partition of the index set $\{1,2,\ldots ,n\}$ into two disjoint subsets $S$ and $Z$ such that
$S\cup Z = \{1,2,\ldots ,n\}$ and $S\cap Z = \emptyset$. In particular, for any $x\in \mathbb{R}^n$, the partition
$(S(x),Z(x))$ refers to the support $S({x})$ of $x$ and its complement -- the zero set
$Z(x)$, namely
\begin{equation}
S(x) = \{i:{x}_i\neq 0, 1\leq i\leq n\},\ \ Z(x) = \{i:{x}_i =0, 1\leq i\leq n\}.
\end{equation}
\end{Def}

\begin{Def}[$k$-balancedness]\label{sec:two-KBdef}
A subspace $V\subseteq\mathbb{R}^n$ is $k$-balanced (in $l_1$ norm) if for any partition $(S,Z)$ with cardinality of $S$ equals to $k$
\begin{equation*}
\|v_S\|_1\leq\|v_Z\|_1, \forall v\in V.
\end{equation*}
It is strictly $k$-balanced if the strict inequality holds for all $v\neq 0$.
\end{Def}

Definitions of $k$-balancedness was introduced by Zhang in \cite{zhang05-2}.
However, $k$-balancedness was used by Donoho and Huo in \cite{donoho-huo-2001},
Elad and Bruckshtein in \cite{elad-bruckshtein-02}, Gribnoval and Nielsen  in \cite{gribnoval-nielsen-03} .

\begin{theorem}[Necessary and Sufficient Conditions for Recovery]\label{sec:ch2-NSF}
Let $A\in\mathbb{R}^{m\times n}$ and $B\in\mathbb{R}^{p\times n}$ be full rank such that $p+m=n$ and $AB^T = 0$.
Then for any $\tilde{x}$ with $\|\tilde{x}\|_0\leq k$ and $\tilde{y} = A\tilde{x}$, $\tilde{x}$
uniquely solves (\ref{intro-0P}) and (\ref{intro-LP}) if and only if $range(B^T)\subset\mathbb{R}^n$ is strictly $k$-balanced.
\end{theorem}

In \cite{zhang05-2} Zhang stated Theorem \ref{sec:ch2-NSF} in its current form and gave a simple proof
by connecting equivalent recoverability conditions for different spaces.
The theorem was used without being stated explicitly by Donoho and Huo in \cite{donoho-huo-2001}
and by Elad and Bruckshtein in \cite{elad-bruckshtein-02} and was stated as Lemma
by Gribnoval and Nielsen in \cite{gribnoval-nielsen-03}.

\begin{Def}[Dictionary]
We say that $A$ is a dictionary if the columns of $A$ are unit vectors.
\end{Def}

\begin{Def}[Coherence of a Dictionary]
Let $A\in\mathbb{R}^{m\times n}$ be a dictionary. The coherence of a dictionary $M(A)$ is defined by
\begin{equation}
M(A) = \max_{i\neq j}|\langle a_i,a_j\rangle|,
\end{equation}
where $a_i,1\leq i\leq n$, is the $i$-th column of $A$.
\end{Def}

Next theorem is due to Gribnoval and Nielsen \cite{gribnoval-nielsen-03}.

\begin{theorem}\label{sec:ch2-GN-1}
Let $k$ be a natural number and let $\|\tilde{x}\|_0 \leq k$.
For any dictionary $A$, if $k < \frac{1}{2}\left(1+ \frac{1}{M(A)}\right)$ and $\tilde{y} = A\tilde{x}$, then $\tilde{x}$ is the unique solution
to both (\ref{intro-0P}) and (\ref{intro-LP}).
\end{theorem}

It is necessary to mention that if $m$ is a power of 2, then there exists a dictionary $A$ such that $M(A) = \frac{1}{\sqrt{m}}$.
See, for example  \cite{calderbank-cameron-97} and \cite{strohmer-heath-03}.

\section{Main Result}

\begin{Def}[$\gamma_{1,\infty}$-width]
Let both $A\in\mathbb{R}^{m\times n}$ and $B\in\mathbb{R}^{p\times n}$ be of full rank and $p+m=n$ and $AB^T = 0.$
We define $\gamma_{1,\infty}$-width of $A$ to be 
\begin{equation}
\ds\gamma_{1,\infty}(A) = \min_{x\in\mathbb{R}^{p},\ x\neq 0}\frac{\|B^T x\|_1}{\|B^T x\|_\infty} = \min_{\|B^T x\|_\infty = 1}\|B^T x\|_1.
\end{equation}
\end{Def}

The feasible set $\{x\in\R^p: \|B^T x\|_\infty = 1\}$ is non-convex, however, it is a union of $2n$ convex sets
\begin{equation}
\{x\in\R^p: \|B^T x\|_\infty = 1\} = \bigcup_{i=1}^n\pm F_i,
\end{equation}
where
\begin{equation}
F_i = \{x\in\R^p: [B^T x]_i = 1; |[B^T x]_j| \leq 1, j\neq i\}.
\end{equation}

Therefore,
\begin{equation}
\gamma_{1,\infty}(A) = \min_{1\leq i\leq n}\min_{x\in F_i}\|B^Tx\|_1.
\end{equation}

For every $1\leq i\leq n$, $\min_{x\in F_i}\|B^Tx\|_1$ could be rewritten as a linear program via a standard transformation.
Therefore, in order to compute $\gamma_{1,\infty}(A)$ it is necessary to solve $n$ linear programs.
While it requires considerable computational efforts for a large $n$, the problem is solvable in polynomial time.

Alternatively, one can solve the reciprocal problem
\begin{equation}
\ds\frac{1}{\gamma_{1,\infty}(A)} = \max_{\|B^T x\|_1= 1}\|B^T x\|_{\infty}.
\end{equation}

Next proposition presents sufficient condition for recovery.
It follows directly from Theorem \ref{sec:ch2-NSF}.

\begin{proposition}[Sufficient Condition for Recovery]\label{sec:three-WSCR}
Recovery is guaranteed whenever $k<\frac{1}{2}\gamma_{1,\infty}(A)$.
\end{proposition}

\begin{proof}
Note, that $\|v_S\|_1 \leq \|v_Z\|_1$ is equivalent to $\ds\|v_S\|_1 \leq \frac{1}{2}\|v\|_1$. Therefore,
\begin{equation*}
\|v_S\|_1 \leq k\|v_S\|_{\infty} \leq k\|v\|_{\infty} < \frac{1}{2}\gamma_{1,\infty}(A)\|v\|_{\infty} \leq \frac{1}{2}\frac{\|v\|_1}{\|v\|_{\infty}}\|v\|_{\infty} \leq \frac{1}{2}\|v\|_1.
\end{equation*}
\end{proof}


Now we are ready to show that estimated sparsity $k$ for guaranteed recovery of a dictionary $A$ computed using $\gamma_{1,\infty}$-width of $A$ is always greater or equal to the estimated sparsity $k$ computed using coherence of $A$.

\begin{theorem}\label{sec:three-Lemma7}
Let $A\in\mathbb{R}^{m\times n}$ be a dictionary and $m<n$. Let $B\in\mathbb{R}^{p\times n}$,
such that $p+m = n$ and $AB^T = 0$.
Let $k_1$ and $k_2$ be the sparsities for guaranteed recovery estimated by
$\gamma_{1,\infty}(A)$ and $M(A)$ respectively. Then $k_1\geq k_2$.
\end{theorem}

\begin{proof}
According to Theorem~\ref{sec:ch2-GN-1} and Proposition~\ref{sec:three-WSCR}, it is enough to show that
\begin{equation}
1 + \frac{1}{M(A)} \leq \gamma_{1,\infty}(A).
\end{equation}

We will follow the proof of Gribnoval and Nielsen \cite{gribnoval-nielsen-03}.

Let $v~\in~range(B^T)$, then $Av = 0$, or, in vector form $\sum_{i=1}^n v_i a_i = 0$, where $a_i$, $1\leq i\leq n$ is the $i$-th column of $A$.
Then, $v_1 a_1 = -\sum_{i=2}^n v_i a_i$. Taking the inner product of both sides with $a_1$, we get
$v_1 = -\sum_{i=2}^n v_i\langle a_i,a_1\rangle$. It follows that
\begin{equation}
|v_1| = \left|-\sum_{i=2}^n v_i\langle a_i,a_1\rangle\right| \leq M(A)\sum_{i=2}^n |v_i| = M(A)(\|v\|_1 - |v_1|),
\end{equation}
or
\begin{equation}
|v_1|(1+M(A)) \leq \|v\|_1 M(A).
\end{equation}
The same way for $2\leq i \leq n$, we get
\begin{equation}
|v_i|(1+M(A)) \leq \|v\|_1 M(A).
\end{equation}

Since this is true for every index $1\leq i \leq n$, it follows that for every vector $v~\in~range(B^T)$ the following inequality holds:

\begin{equation}
1 + \frac{1}{M(A)} \leq \frac{\|v\|_1}{\|v\|_{\infty}}.
\end{equation}

Now if we take minimum over all $v~\in~range(B^T)$ we get:
\begin{equation}
1 + \frac{1}{M(A)} \leq\min_v\frac{\|v\|_1}{\|v\|_{\infty}} =\gamma_{1,\infty}(A).
\end{equation}
which completes the proof.
\end{proof}

\section{Conclusion}
In this paper we defined $\gamma_{1,\infty}$-width of a measurement matrix $A$ and showed that if $A$ is a dictionary,
our approach to estimate recoverability properties of $A$ is at least as good as coherence approach.
Moreover, our method can be used to estimate the recoverability of $A$ even in the case $A$
is not a dictionary. Currently, to the best of our knowledge, no other such method exists.\\


\begin{thebibliography}{10}

\bibitem{calderbank-cameron-97}
A. R. Calderbank, P. J. Cameron, W. M. Kantor and J. J. Seidel.
\newblock $Z_4$-Kerdock Codes, Orthogonal Spreads, and Extremal Euclidean Line-sets.
\newblock Proc. London Math. Soc. (3), no. 2, pp. 436-480, 1997.

\bibitem{donoho-huo-2001}
D. Donoho and X. Huo.
\newblock Uncertainty principles and ideal atomic decompositions.
\newblock IEEE Trans. Inf. Theory 47 (2001), 2845-2862.

\bibitem{elad-06}
M. Elad.
\newblock Optimized Projections for Compressed Sensing
\newblock IEEE Trans. Sign. Processing 55 (2002), 5695 - 5702.

\bibitem{elad-bruckshtein-02}
M.Elad and A. Bruckshtein.
\newblock A Generalized Uncertainty Principle and Sparse Representations in Pairs of Bases
\newblock IEEE Trans. Inf. Theory 48 (2002), 2558-2567.

\bibitem{gribnoval-nielsen-03}
R\`{emi} Gribnoval and Morten Nielsen.
\newblock Sparse Representation in Union of Bases.
\newblock IEEE Trans. Inf. Theory 49 (2003), no. 12, 3320-3325.


\bibitem{strohmer-heath-03}
T. Strohmer and R. Heath.
\newblock Grassmannian Frames with Applications to Coding and Communications.
\newblock Appl. Comp. Harm. Anal., vol. 14, no. 3, pp. 257-275, 2003.

\bibitem{zhang05-2}
Y. Zhang.
\newblock  A Simple Proof for Recoverability of $l_1$ Minimization:
Go Over or Under?
\newblock Technical report TR05-09, Department of Computational and
Applied Mathematics, Rice University, Houston, TX, 2005.


\end{thebibliography}
\bibliographystyle{plain}

\end{document}